\documentclass[conference]{IEEEtran}

\ifCLASSINFOpdf
\else
\fi

\usepackage[cmex10]{amsmath}
\usepackage{amsfonts}
\usepackage{amssymb}
\usepackage{graphicx}%

\usepackage{color}

\newtheorem{example}{Example}
\newtheorem{remark}{Remark}
\newtheorem{definition}{Definition}
\newtheorem{theorem}{Theorem}
\newtheorem{corollary}{Corollary}
\newtheorem{lemma}{Lemma}
\newtheorem{proposition}{Proposition}
\newenvironment{proof}[1][Proof]{\noindent\textbf{#1.} }{\ \rule{0.5em}{0.5em}}

\begin{document}

\title{Bounds for complexity of syndrome decoding for poset metrics}
\author{\IEEEauthorblockN{Marcelo Firer}
\IEEEauthorblockA{Institute of Mathematics, Statistics \\ and Scientific Computing \\
University of Campinas\\
S\~ao Paulo, Brazil 13083-859\\
Email: mfirer@ime.unicamp.br}
\and
\IEEEauthorblockN{Jerry Anderson Pinheiro}
\IEEEauthorblockA{Institute of Mathematics, Statistics \\ and Scientific Computing \\
University of Campinas\\
S\~ao Paulo, Brazil 13083-859\\
Email: jerryapinheiro@gmail.com}
}
\maketitle

\begin{abstract}
In this work we show how to decompose a linear code relatively to any given poset metric. We prove that the complexity of syndrome decoding is determined by a maximal (primary) such decomposition and then show that a refinement of a partial order leads to a refinement of the primary decomposition. Using this and considering already known results about hierarchical posets, we can establish upper and lower bounds for the complexity of syndrome decoding relatively to a poset metric.
\end{abstract}
\IEEEpeerreviewmaketitle
\section{Introduction}

The \textit{poset metrics}, which are metrics determined by partial orders over a finite set, were introduced in the context of coding theory by
Brualdi et al. \cite{brualdi} in 1995. The interest in such metrics may arise
in the modeling process of the decoding-decision criteria over some channels, as in  \cite{barg1}, for the usage to perform
unequal error protection, as proposed in \cite{flp}, or the use,
combined with block-codes, in steganography, as in \cite{darriti}.

Poset metrics generalize the Hamming metric and include the case of
Niederreiter-Rosenbloom-Tsfasman (NRT) metrics. Moreover, those two examples
may be seen as extreme poset-metrics. While the Hamming metrics may be considered as a discrete
generalization of an Euclidean metric, the NRT metric with one chain is an ultra-metric,
having some properties that defy the Euclidean intuition. The hierarchical
poset metrics (which include both the cases) are well understood, with explicit formulas being given for
the most important invariants of coding theory (minimal distance, packing
radius, Singleton bound, characterization of perfect codes), presented in 
\cite{felix}. The structure of general NRT metrics is somehow understood as we can see in \cite{Muniz} and \cite{bargPruna}  or \cite{PMunizFirer}, for the case of a single chain.

When considering
hierarchical poset metrics, in \cite{felix}, it was 
described how the complexity of syndrome decoding may be significantly reduced by considering the poset-metric structure and working with a poset that is poset-metric equivalent to the given one. In the extreme case of NRT metric with one chain, syndrome decoding becomes
just a linear map. In this work we extend the results of \cite{felix} by
establishing bounds for complexity of syndrome
decoding of linear codes, considering decoding to be performed
with the use of a poset metric. The bounds are established through a
characterization of what would be a primary decomposition of a code (related to
a poset).

\section{Posets and poset-metrics}

 Let us denote $\left[n\right]  =\left\{1,2,\ldots,n\right\}$ and let $\preceq$ be a partial order
relation defined on $\left[ n\right]$. The pair $P=\left(  \left[
n\right],\preceq\right)=\left(  \left[
n\right],\preceq_P\right)$ is called a (finite) \emph{partially ordered
set} (abbreviated as \emph{poset}). An ideal in $P$ is a subset $I\subseteq P$
with the property that if $i\in I$ and $j\preceq i$, than $j\in I$. Given a
subset $X\subset I$, we denote by $\left\langle X\right\rangle$ the smallest
ideal containing $X$, called the \emph{ideal generated by} $X$.

A \emph{chain of length} $k$ in a poset $P=\left(  \left[  n\right]
,\preceq\right)$ is a labelled family $\{i_{1},\ldots,i_{k}\}\in\left[  n\right]$,
such that $i_{l}\prec i_{l+1}$ for every $l\in\{1,\ldots,k-1\}$. The chain is said to
be \emph{saturated} if every $i_{l+1}$ \emph{covers} $i_{l}$, in other words, there is no $m\in [n]$ such that $i_l \prec m \prec i_{l+1}$. The
\emph{height }$h\left(  i\right)$ in $P$ is the maximal length of a
saturated chain that has $i$ as a maximal element and the \emph{height} $h:=h\left(  P\right)  $ is the maximal height of its elements. 
We denote 
\[
H_{i}:=\left\{  j\in\left[  n\right]  |h\left(  j\right)  =i\right\}
\]
an call it the $i$\emph{-th level of }$\left[  n\right]  $. It is easy to see
that $\left[  n\right]  $ has a \textit{level decomposition} obtained by disjoint union of levels
\[
\left[  n\right]  =\bigcup_{i=1,\ldots,h}^\circ H_{i}\text{.}%
\]
In this case, we say that $P$ has (height) structure \emph{of type} $\left(
n_{1},\ldots,n_{h}\right)$, where $n_{i}=\left\vert H_{i}\right\vert$. Given an integer $h \leq n$, if we denote
\[
\mathcal{P}\left(  n;n_{1},\ldots,n_{h}\right)  =\left\{  P\in\mathcal{P}\left(
n\right)  |P\text{ has type }\left(  n_{1},\ldots,n_{h}\right)\right\},
\]
we have a partition of $\mathcal{P}\left(  n\right)  $, the set of all posets on $[n]$, into different
\emph{height components} (or just $h$-\emph{components}). Obviously, $\mathcal{P}\left(  n;n_{1},\ldots,n_{h}\right)\neq \emptyset$ \textit{iff} $n=n_1+\cdots +n_h$.

Given two posets $P=\left(  \left[  m\right]  ,\preceq_{P}\right)  $ and
$Q=\left(  \left[  n\right]  ,\preceq_{Q}\right)  $, an \emph{order
homomorphism} from $P$ into $Q$ is a map $\phi:\left[  m\right]
\rightarrow\left[  n\right]  $ such that $i\preceq_{P}j$ implies $\phi\left(
i\right)  \preceq_{Q}\phi\left(  j\right)  $. An isomorphism of order is a
bijective map $\phi$ such that both $\phi$ and its inverse $\phi^{-1}$ are
homomorphisms, or equivalently, a bijection \ $\phi:\left[  n\right]
\rightarrow\left[  n\right]  $ such that $i\preceq_{P}j$ if and only if
$\phi\left(  i\right)  \preceq_{Q}\phi\left(  j\right)  $. In case $P=Q$, we
say that an isomorphism $\phi:\left[  m\right]  \rightarrow\left[  n\right]  $
is an automorphism. The group of automorphisms of $P$ is denoted by
$\mathrm{Aut}_{P}$.

Given two posets $P=\left(  \left[  n\right]  ,\preceq_{P}\right)$ and
$Q=\left(  \left[  n\right]  ,\preceq_{Q}\right)$ over the same set $\left[
n\right] $, we say that $P$ \emph{is finer (or smaller) then} $Q$ (and
write $P\leq Q$) if $i\preceq_{P}j$ implies $i\preceq_{Q}j$. Equivalently,
$P\leq Q$ if and only if the identity map is an order homomorphism from $P$ into
$Q$. With this relation, the set $\mathcal{P}$ of partial orders over
$\left[  n\right],$%
\[
\mathcal{P}:=\mathcal{P}_{n}=\left\{  P=\left(  \left[  n\right]  ,\preceq
_{P}\right)  ,\leq\right\}
\]
is itself a partially ordered set. The trivial order ($i\preceq j$ if, and only if, $i=j$) is
the (unique) minimal element in $\mathcal{P}$. The $n!$ linear orders
$\sigma (1)\preceq \sigma (2)\preceq\cdots\preceq \sigma (n)$ (where $\sigma $ is a permutation of $[n]$) are the
(isomorphic) maximal elements in $\mathcal{P}$. Since a linear order defines a unique
chain, a linear order is also known as \emph{chain }and the trivial order may
be called an \emph{anti-chain}.

\begin{example}
(Hierarchical order) Let $P$ be a poset having level decomposition
$\left[  n\right]  =\bigcup_{i=1}^h H_{i}$ with $n_{i}%
=\left\vert H_{i}\right\vert $.  We say that the poset $P$ is \emph{hierarchical of
type }$(n_{1},n_{2},\ldots,n_{h})$ if $i\preceq j$ only when 
$i=j$ or $i\in H_{l_{i}},j\in H_{l_{j}}$ and $l_{i}<l_{j}$. Note that when $h=1$, we have a trivial (and minimal) order and when $h=n$, we have a linear (and maximal) order.
\end{example}


If we consider an $h$-component $\mathcal{P}\left(  n;n_{1},\ldots,n_{h}\right)
$ of $\mathcal{P}$, any of its maximal element is a hierarchical poset and all
those maximal posets are isomorphic.

\subsection{Poset metrics}

Let $\mathbb{F}_{q}$ be a finite field with $q$ elements and $\mathbb{F} 
_{q}^{n}$ the $n$-dimensional vector space of $n$-tuples over $\mathbb{F}_{q}$. 
Given $x=\left( x_{1},x_{2},\ldots,x_{n}\right)  \in\mathbb{F}_{q}^{n}$ , the
\textit{support} of $x$ is the set $\mathrm{supp}(x)=\{i\in [n]|x_{i}\neq0\}$. Let $P=\left(  \left[  n\right],\preceq\right)$ be a poset and define the
$P$-weight of $x$ as
\[
\omega_{P}(x):=\left\vert \langle\mathrm{supp}(x)\rangle\right\vert
\]
where $\left\vert \cdot\right\vert $ denotes the cardinality of the given set.

The function $d_{P}:\mathbb{F}_{q}^{n}\times\mathbb{F}_{q}^{n}\longrightarrow
\mathbb{N}$, defined by
\[
d_{P}(x,y)=\omega_{P}(x-y)
\]
is a metric on $\mathbb{F}_{q}^{n}$ called a \textit{poset-metric} or
$P$\textit{-metric}, endowed with such a metric, $\mathbb{F}_{q}^{n}$ is
called a \emph{poset space} or a $P$\emph{-space}.
%

\begin{example}
(Hierarchical poset) Let $P$ be a hierarchical poset of type $\left(  n_{1},n_{2},\ldots,n_{h}\right)$. Given
$x\in\mathbb{F}_{q}^{n}$ let $m=m\left(  x\right)  =\max\left\{
i|\mathrm{supp}(x)\cap H_{i}\neq\emptyset\right\}$. Then, the $P$-weight of
$x$ is obtained by considering the Hamming weight of the intersection
$\mathrm{supp}(x)\cap H_{m}$ to which we need to add the size of the previous
levels in the hierarchy, that is
\[
\omega_{P}\left(  x\right)  =\left\vert \mathrm{supp}(x)\cap H_{m}\right\vert
+\left(  n_{1}+n_{2}+\cdots+n_{m-1}\right)  \text{.}
\]
As a particular case, when $h(P)$ is minimal ($h(P)=1$), we have that $|\langle\mathrm{supp}(x)\rangle|=|\mathrm{supp}(x)|$ hence $\omega_{P}$ and $d_{P}$ are just the usual Hamming weight
and distance. When $h(P)$ is maximal ($h(P)=n$),
$P$ is a linear order and assuming $1\preceq2\preceq
\cdots\preceq n$ we have that  $\omega_{P}\left(  x\right)  =\max\left\{  i|x_{i} \neq0\right\}$ and $d_{P}\left(  x,y\right)  =\max\left\{  i|x_{i}\neq y_{i}\right\}$.

\end{example}

%

%

We remark that the family of hierarchical posets, in which we are concerned
in this work, is a large family of posets. Indeed, to define a hierarchical
poset on $[n]$ is equivalent, up to isomorphism,  to define a positive partition of $n$ and the
number $p(n)$ of such partitions behaves asymptotically as $\frac{1}{4n\sqrt{3}}e^{\pi\sqrt{2n/3}}$ (see \cite{hardy}). 

\subsection{Linear Isometries and equivalence of codes}

Our primary concern is with linear codes, so our appropriate equivalence will be dictated by linear isometries. Let us denote $G_{P}:=GL_{P}\left(  n\right)  =\left\{  T\in GL\left(  \mathbb{F}_{q}^{n}\right)
|T\text{ is a }P\text{-isometry}\right\}$, the group of linear
isometries of $\left(  \mathbb{F}_{q}^{n},d_{P}\right)$, where $GL\left(
\mathbb{F}_{q}^{n}\right)$ is the group of invertible linear transformations
of $\mathbb{F}_{q}^{n}$. 


\begin{remark}
\label{iso} From \cite{groups} we have that $GL_{P}\left(  n\right)  =S_{P}\ltimes\Delta_{P}$,
where $S_{P}$ is the permutation of coordinates corresponding to
automorphisms of the poset $P$, and $\Delta_{P}$ is the group of invertible
triangular matrices such that entry $ij$ is zero if 
$i\npreceq_{P}j$. If $P \leq Q$ then
$\Delta_{P}\subseteq\Delta_{Q}$ but we do not necessarily have $S_{Q}\subseteq
S_{P}$. We denote by $Aut\left(  P\right)$ the group of order
automorphisms of $P$. Given $\sigma\in Aut\left(  P\right)$, 
$T_{\sigma}(x)=(x_{\sigma (1)},\ldots ,x_{\sigma (n)})$ is a $P$-isometry.
\end{remark}

An $\left[  n,k,\delta_{P}\right]  _{q}\ \ $\emph{poset code} is a
$k$-dimensional subspace $C\subset\mathbb{F}_{q}^{n}$, where
$\mathbb{F}_{q}^{n}$ is considered to be a $P$-space, and $\delta
_{P}=\delta_{P}(C)=\min\{\omega_{P}(x)|0\not =x\in C\}$ is
the \emph{minimal} $P$\emph{-distance} of the code $C$.
When $\delta_{P}$ does not play a relevant role, we will write
just $\left[n,k\right]_{q}$.

An $\left[n,k,\delta_{P}\right]  _{q}$ code $C\subset
(\mathbb{F}_{q}^{n},d_{P})$ and an $\left[n,k,\delta_{P^{\prime}}\right]
_{q}$ code $C^{\prime}\subset(\mathbb{F}_{q}^{n},d_{P^{\prime}})$
are said to be \emph{equivalent} if there is a linear isometry $T:(\mathbb{F}%
_{q}^{n},d_{P})\longrightarrow(\mathbb{F}_{q}^{n},d_{P^{\prime}})$ such that
$T(C)=C^{\prime}$. In particular we must have $P$ and
$P^{\prime}$ order-isomorphic. We allow $P$ and $P^{\prime}$ to be different
but isomorphic posets to allow different labellings. If we consider, for
example, the posets $P$ and $P^{\prime}$ to be generated by the relations
$1\preceq_{P}2\preceq_{P}\cdots\preceq_{P}n$ and $n\preceq_{P^{\prime}%
}n-1\preceq_{P^{\prime}}\cdots\preceq_{P^{\prime}}1$, then we have that the
one-dimensional codes $C$ and $C^{\prime}$, generated respectively by the
canonical vectors $e_{n}$ and $e_{1}$ are equivalent when considering the poset metric $d_{P}$
and $d_{P^{^{\prime}}}$ to evaluate $C$ and $C^{\prime}$ respectively.

We should remark that up to order-isomorphism a hierarchical poset is
determined by the values of $\left(  n;n_{1},\ldots,n_{h}\right)  $. Hence, we
may assume a \emph{natural labelling} of the poset:  $i\preceq j$ implies $i\leq j$. In our case it means that $H_{i}=\{ n_1+\cdots + n_{i-1}+1,\ldots,n_{1}+\cdots +n_{i-1}+n_{i}\}$ for every $i\in  [h(P)]$.

\section{Partitions and Decompositions}

Let us consider a partition $J=\cup_{i=1}^{r}J_{i}$ of a subset $J\subseteq\left[  n\right]$. We
will denote this partition by $\mathcal{J}=(  J;J_{i})_{i=1}^r$. If we write
$J_{0}=\left[  n\right]  \backslash J=\left\{  i\in\left[  n\right]  ;i\notin
J\right\}  $ we get a \emph{pointed partition} $\mathcal{J}^{\ast}=\left(
J;J_0;J_{i}\right)_{i=1}^r=\left(J_0;J_{i}\right)_{i=1}^r$ where $J_{0}=\emptyset$ \emph{iff
}$J=\left[  n\right]$. We remark that $J_{0}$ has a
special role, since it is the unique part we allow to be empty. From now on, we consider only pointed partitions, so we omit the symbol $^{\ast}$ and the adjective ``pointed". A
 partition $\mathcal{J}$  can be
refined in two ways, either by increasing the number of parts or by enlarging
the distinguished part $J_{0}$. We remark that excepts for the pointer
$J_{0}$, the order of the other parts is irrelevant, so that
\[
\left(  J_{0};\left\{  1,2\right\}  ,\left\{  3,4,5\right\}  \right)  =\left(
J_{0};\left\{  5,4,3\right\}  ,\left\{  1,2\right\}  \right)  \text{.}%
\]

\begin{definition}
An $l$\emph{-split of a  partition} $\mathcal{J}=\left(  J_{0}%
;J_{i}\right)  _{i=1}^{r}$ is a  partition $\mathcal{J}^{\prime
}=\left(  J_{0};J_{i}^{\prime}\right)  _{i=1}^{r+1}$ where $J_{i}%
=J_{i}^{\prime}$ for $i\neq l$ and $J_{l}=J_{l}^{\prime}\cup J_{r+1}%
^{\prime}$, where both $J_{l}^{\prime}$ and $J_{r+1}^{\prime}$  are non-empty.
It means that we have just split $J_{l}$ into two components. An $l$%
\emph{-aggregate }of a partition $\mathcal{J}=\left(  J_{0};J_{i}\right)
_{i=1}^{r}$ is a partition $\mathcal{J}^{\prime}=\left(  J_{0}^{\prime}%
;J_{i}^{\prime}\right)  _{i=1}^{r}$ where $J_{i}=J_{i}^{\prime}$ for $i\neq
l,0$ and $J_{l}=J_{l}^{\prime}\cup J_{l}^{\ast}$, and $J_{0}^{\prime}=J_{0}\cup
J_{l}^{\ast}$ for some $\emptyset\neq J_{l}^{\ast}\subset J_{l}$, that is,
some elements of $J_{l}$ are aggregated into the distinguished part $J_{0}$.
\end{definition}

\begin{definition}
We say that a partition $\mathcal{J}^{\prime}$ is a $1$\emph{-step refinement
of} $\mathcal{J\,}$\ if $\mathcal{J}^{\prime}$ is obtained from $\mathcal{J}$
by performing a single $l$-split or a single $l$-aggregate operation, for some
$l<\left\vert \mathcal{J}\right\vert $. $\mathcal{J}^{\prime}$ is a \emph{
refinement of} $\mathcal{J\,}$\ if $\mathcal{J}^{\prime}$ can be obtained
from $\mathcal{J}$ by a successive number of $1$-step refinements. We use the
notation $\mathcal{J}^{\prime}\leq\mathcal{J}$ to denote a refinement and
$\mathcal{J}^{\prime}\leq_{l}\mathcal{J}$ to denote that $\mathcal{J}^{\prime
}$ is a $1$-step refinement of $\mathcal{J}$ performed by an $l$-split or
$l$-aggregate. In case the operation is important to be specified, we will use
the notation $\mathcal{J}^{\prime}\leq_{l}^{s}\mathcal{J}$ and $\mathcal{J}%
^{\prime}\leq_{l}^{a}\mathcal{J}$ for a $l$-split or a $l$-aggregate, respectively.
\end{definition}

\begin{example}
The partition $([4];\emptyset ; \{1,2,3,4\})$ can be refined to the partition $([4];\{1,3\};\{2\},\{4\})$, indeed,
\begin{align*}
\left(  \emptyset;\left\{  1,2,3,4\right\}  \right)  & \geq_{1}^{a}\left( 
\left\{  3\right\}  ;\left\{  1,2,4\right\}  \right) \\ & \geq_{1}^{a}\left(
\left\{  1,3\right\}  ;\left\{  2,4\right\}  \right)  \geq_{1}^{s}\left(
\left\{  1,3\right\}  ;\left\{  2\right\}  ,\left\{  4\right\}  \right).
\end{align*}
\end{example}

As we will see now, a decomposition of a linear code $C\subseteq \mathbb{F}_{q}^{n}$ is the algebraic
equivalent of the set partition of $\left[  n\right]  $.

\begin{definition}
\label{code_decomposition}We say that $\mathcal{C=}\left(  C;C_{0}%
;C_{i}\right)  _{i=1}^{r}$ is a \emph{decomposition of} $C$ if each
$C_{i}$ is a subspace of $\mathbb{F}_{q}^{n}$ such that:

\begin{itemize}
\item $C=\oplus_{i=1}^{r}C_{i}$ with $\dim\left(C_{i}\right)  >0$ for every $i\in\{1,\ldots,r\}$;
\item $C_{0}=\left\{  \left(  x\,_{1},x_{2},\ldots,x_{n}\right)  |x_{i}=0\text{ if }i\in\mathrm{supp}\left(  C\right)  \right\}$;
\item $\left(\mathrm{supp}\left(  C_{0}\right)  ;\mathrm{supp}\left(
C_{i}\right)  \right)  _{i=1}^{r}$ is a  pointed partition, where
\end{itemize}
\begin{equation*}
\mathrm{supp}\left(  C_{i}\right)   =\left\{  j\in\left[  n\right]
|x_{j}\neq0\text{ for some }\left(  x_{1},\ldots,x_{n}\right)  \in
C_{i}\right\} .
\end{equation*}
\end{definition}

\begin{definition}
An $l$-split, $l$-aggregate, $1$-step refinement and a refinement $\mathcal{C}%
^{\prime}=\left(  C;C_{0}^{\prime};C_{i}^{\prime}\right)  _{i=1}^{r\prime}$ of
a decomposition $\mathcal{C=}\left(  C;C_{0};C_{i}\right)  _{i=1}^{r}$ are
defined according to
\[
\left(  \mathrm{supp}\left(  C_{0}^{\prime}\right)  ;\mathrm{supp}\left(
C_{i}^{\prime}\right)  \right)  _{i=1}^{r\prime}%
\]
being a $l$-split, $l$-aggregate, $1$-step refinement or a refinement of
$\left(  \mathrm{supp}\left(  C_{0}\right)  ;\mathrm{supp}\left(
C_{i}\right)  \right)  _{i=1}^{r}$.
\end{definition}

Fixing the canonical base $\beta=\left\{  e_{1},\ldots,e_{n}\right\}  $ of $\mathbb{F}_q^n$, and given $I\subset\left[
n\right]$, the $I$-\emph{coordinate subspace} $V_I$ is    
\[
V_I=\mathrm{span}(\left\{e_i |i\in I\right\})=\left\{  \sum_{i\in I}x_{i}e_{i}|x_i\in \mathbb{F}_q\right\} .
\]

It is clear that a code $C$ admits a \emph{maximal} decomposition (a decomposition that does not admit a refinement)
$(C;C_0;C_i)_{i=1}^r$ and then we say that $C$ has $r$ components. Given
such a decomposition, we define
\begin{center}
$
V_{i}:= V_{\textrm{supp}(C_i)}=\mathrm{span}(\left\{e_i|i\in \mathrm{supp}\left(
C_{i}\right)\right\} ) $
\end{center}
to be the \emph{support-space of (the component) }$C_{i}$. We consider
$\left[  n\right]  _{C}=$\textrm{supp}$\left(  C\right)  $ and $\left[
n\right]  ^{C}=\left[  n\right]  \backslash\left[  n\right]  _{C}$. In the case that $\left[  n\right]  ^{C}\neq\emptyset$ we  write 
$V_{0}=\{  \sum_{j\in\left[  n\right]  ^{C}}\lambda_{j}e_{j}|\lambda
_{j}\in\mathbb{F}_{q}\}$ and denote $C_{0}=V_{0}$. We
say that the decomposition $(C;C_0;C_i)_{i=1}^r$ is supported by the
\emph{environment decomposition }$(V_0;V_i)_{i=1}^r$. In case $\left[
n\right]  ^{C}=\emptyset$, we have $V_{0}=C_{0}=\left\{  0\right\}  $.

Until now, the metric $d_{P}$ to be considered  played no
role in the decomposition of a code. We will introduce now a decomposition
that depends of the poset $P$. Given a linear code $C$, we denote by
$G_{P}\left(  C\right)  $ its orbit under $G_{P}$. Since $G_{P}$ is a group,
we have that $G_{P}\left(  C\right)  =G_{P}\left(  C^{\prime}\right)  $
\emph{iff} there is $T\in G_{P}$ such that $T\left(  C\right)  =C^{\prime}$.
We remark that orbits are equivalence classes, so we write $C\sim_{P}%
C^{\prime}$ if $G_{P}\left(  C\right)  =G_{P}\left(  C^{\prime}\right)  $.

\begin{definition}
\label{P_decomposition} A $P$-\emph{decomposition of} $\emph{C}$ is a decomposition 
$\mathcal{C}_{P}^{\prime}=\left(  C^{\prime};C_{0}^{\prime};C_{i}^{\prime
}\right)  _{i=1}^{r}$ of
$C^{\prime}$ (as in Definition \ref{code_decomposition}) where $C^{\prime}\sim_{P}C$. Each $C_{i}^{\prime}$ is called a \emph{component }of the decomposition. A 
\emph{ trivial }$P$\emph{-decomposition} of $C$ is either the decomposition
$\left(  C;C_{0};C\right)  $ or any $P$-decomposition with a unique factor
$\left(  C^{\prime};C_{0}^{\prime};C^{\prime}\right)  $ where $\left\vert
\mathrm{supp}\left(  C^{\prime}\right)  \right\vert =\left\vert \mathrm{supp}%
\left(  C\right)  \right\vert $ and $\left\vert \mathrm{supp}\left(
C_{0}^{\prime}\right)  \right\vert =\left\vert \mathrm{supp}\left(
C_{0}\right)  \right\vert $.
\end{definition}

\begin{definition}
A code $C$ is said to be $P$\emph{-irreducible} if it does not admit a
non-trivial $P$-decomposition.
\end{definition}

Given a $P$-decomposition $\mathcal{C}^{\prime
}\mathcal{=}\left(  C^{\prime};C_{0}^{\prime};C_{i}^{\prime}\right)
_{i=1}^{r}$  of $C$, we have that if $C_{i}^\prime=\left\{  0\right\}  $ then $ i=0$ and  
$\left[  n\right]  =\cup_{i=0}^{r}$\textrm{supp}$\left(  V_{i}\right)  $. Also, if each $C^\prime_i$ is $P$-irreducible, putting $n_{i}=\dim V_{i}$
and $k_{i}=\dim C_{i}^{\prime}$, it is clear that
\begin{equation*}
\sum_{i=0}^{r}n_{i}   =n=\dim \mathbb{F}_q^n \ \text{ and } \ 
\sum_{i=1}^{r}k_{i}=k=\dim C.
\end{equation*}

Let us consider two $P$-decompositions $\mathcal{C}^{\prime}=\left(
C^{\prime};C_{0}^{\prime};C_{i}^{\prime}\right)  _{i=1}^{r\prime}$ and
$\mathcal{C}^{\prime\prime}=\left(  C^{\prime\prime};C_{0}^{\prime\prime
};C_{i}^{\prime\prime}\right)  _{i=1}^{r^{\prime\prime}}$ of a given code. Associated to those
$P$-decompositions there are two  partitions of $\left[  n\right]  $,
namely $\left(  \mathrm{supp}\left(  C_{0}^{\prime}\right)  ;\mathrm{supp}%
\left(  C_{i}^{\prime}\right)  \right)  _{i=1}^{r\prime}$ and $\left(
\mathrm{supp}\left(  C_{0}^{\prime\prime}\right)  ;\mathrm{supp}\left(
C_{i}^{\prime\prime}\right)  \right)  _{i=1}^{r^{\prime\prime}}$ respectively.
By the definition of a $P$-decomposition there are isometries $T^{\prime
},T^{\prime\prime}\in G_{P}$ such that $T^{\prime}\left(  C\right)  =
C^{\prime}$ and $T^{\prime\prime}\left(  C\right)  =C^{\prime\prime}$. Let
us denote $T=T^{\prime\prime}\circ\left(  T^{\prime}\right)  ^{-1}$ and
$T^{-1}=T^{\prime}\circ\left(  T^{\prime\prime}\right)  ^{-1}$. It is known
\cite{groups} that $T$ induces an automorphism of order $\sigma
_{T}:\left[  n\right]  \rightarrow\left[  n\right]  $ and $\sigma_{T}$ induces
a map on the  partition of $\left[  n\right]  $ determined by the
$P$-decomposition $\mathcal{C}^{\prime}$, namely,
\begin{align*}
\sigma_{T}[  \left(  \mathrm{supp}\left(  C_{0}^{\prime}\right)
;\mathrm{supp} \right. & \left.\left(  C_{i}^{\prime}\right)  \right)  _{i=1}^{r\prime}]
= \\
& \left(  \sigma_{T}\left(  \mathrm{supp}\left(  C_{0}^{\prime}\right)
\right)  ;\sigma_{T}\left(  \mathrm{supp}\left(  C_{i}^{\prime}\right)
\right)  \right)  _{i=1}^{r\prime}.
\end{align*}

\begin{definition}
Let $C\subset\mathbb{F}_{q}^{n}$ be a linear code and $P$ a poset on $\left[
n\right]  $. Let $\mathcal{C}^{\prime}=\left(  C^{\prime};C_{0}^{\prime}%
;C_{i}^{\prime}\right)  _{i=1}^{r\prime}$ and $\mathcal{C}^{\prime\prime
}=\left(  C^{\prime\prime};C_{0}^{\prime\prime};C_{i}^{\prime\prime}\right)
_{i=1}^{r\prime\prime}$ be two $P$-decompositions of $C$. We say that
$\mathcal{C}^{\prime}$ is \emph{a }$P$\emph{-refinement} ($1$\emph{-step }%
$P$-\emph{refinement}) of $\mathcal{C}^{\prime\prime}$ if $\sigma_{T}[
\left(  \mathrm{supp}\left(  C_{0}^{\prime}\right)  ;\mathrm{supp}\left(
C_{i}^{\prime}\right)  \right)  _{i=1}^{r\prime}]  $ is a refinement
($1$-step refinement) of the partition $\sigma_{T}[  \left(  \mathrm{supp}\left(
C_{0}^{\prime\prime}\right)  ;\mathrm{supp}\left(  C_{i}^{\prime\prime
}\right)  \right)  _{i=1}^{r^{\prime\prime}}]$. 

\end{definition}

\begin{definition}
A $P$-decomposition $\left(  C;C_0;C_{i}\right)_{i=1}^r$ is said to be \emph{maximal} if each $C_{i}$ is $P$-irreducible.
\end{definition}

Let us consider the anti-chain poset $H$\ and let $\mathbb{F}_{q}%
=\mathbb{F}_{2}$. Then we have that $d_{H}$ is just the usual Hamming metric
and $G_{H}$ is the permutation group $S_{n}$. Given a code $C$, let
$(C;C_0;C_i)_{i=1}^r$ be its maximal decomposition. If a code $C^{\prime}$
is $H$-equivalent to $C$ then there is $T\in G_{H}$ such that $T\left(
C\right)  =C^{\prime}$ and it is easy to see that $(C^\prime;T(C_0);T(C_i))_{i=1}^r$ is a maximal decomposition of
$C^{\prime}$, so that, when considering the Hamming metric, a maximal
decomposition is also a maximal $H$-decomposition. This is not the general case, as we can
see in the following example.

\begin{example}
Let us consider $C$ to be the $1$-dimensional binary code of length $n$ generated by $\left(
1,1,\ldots,1\right)  $, let $H$ be the anti-chain order and $P$ the chain
order defined by the relations $1\preceq2\preceq\cdots\preceq n$. It follows
that $C$ is $H$-irreducible but not $P$-irreducible. Indeed, it is easy to
show that the map
\[
T\left(  x_{1},\ldots,x_{n-1},x_{n}\right)  =\left(  x_{1}-x_{n},\ldots,x_{n-1}%
-x_{n},x_{n}\right)
\]
is a $P$-isometry and $C^{\prime}=T\left(  C\right)  $ is the code
$\mathrm{span}\left(  \left\{  e_{n}\right\}  \right)  $ generated by $e_{n}$,
hence
\[
\mathcal{C}^{\prime}=\left(
\mathrm{span}\left(  \left\{  e_{n}\right\}  \right)  ;\mathrm{span}\left(
\left\{  e_{1},\ldots,e_{n-1}\right\}  \right)  ;\mathrm{span}\left(  \left\{
e_{n}\right\}  \right)  \right)
\]
is a maximal $P$-decomposition of $C$.
\end{example}

\begin{definition}
Let $\mathcal{C}^{\prime}:=\mathcal{C}_{P}^{\prime}=\left(
C^{\prime};C_{0}^{\prime};C_{i}^{\prime}\right)  _{i=1}^{r}$ be a
$P$-refinement of an $\left[  n,k\right]  _{q}$-linear code $C$. The
\emph{profile } of
$\mathcal{C}^{\prime}$ is the array 
\[ \mathrm{profile}\left(  \mathcal{C}^{\prime}\right) :=\left[  \left(  n_{0}^{\prime}%
,k\,_{0}^{\prime}\right)  ,\left(  n_{1}^{\prime},k_{1}^{\prime}\right)
,\ldots,\left(  n_{r}^{\prime},k_{r}^{\prime}\right)  \right]  
\]
where
\[
n_{i}^{\prime}=\left\vert \mathrm{supp}\left(  C_{i}^{\prime}\right)
\right\vert \text{ and }k_{i}^{\prime}=\dim\left(  C_{i}^{\prime}\right)
\text{.}%
\]
\end{definition}
It is obvious that $n=n_{0}+n_{1}+\cdots+n_{r}$ and $k=k_{1}+k_{2}%
+\cdots+k_{r}$.
The following theorem states that the profile of a $P$-maximal decomposition
$\mathcal{C}_{P}^{\prime}$ of a code $C$ depends (essentially) exclusively on
$C$, not on $\mathcal{C}_{P}^{\prime}$.

\begin{theorem}
\label{teo1}Let $C$ be an $\left[  n,k\right]  _{q}$-linear code and let $P$
be a poset over $\left[  n\right]  $. Let $\mathcal{C}^{\prime}$ and
$\mathcal{C}^{\prime\prime}$ be two maximal $P$-decompositions of $C$ with%
$$
\mathrm{profile}\left(  \mathcal{C}^{\prime}\right)    =\left[  \left(
n_{0}^{\prime},k\,_{0}^{\prime}\right)  ,\left(  n_{1}^{\prime},k_{1}^{\prime
}\right)  ,\ldots,\left(  n_{r}^{\prime},k_{r}^{\prime}\right)  \right] 
$$
and
$$
\mathrm{profile}\left(  \mathcal{C}^{\prime\prime}\right)   =\left[
\left(  n_{0}^{\prime\prime},k_{0}^{\prime\prime}\right)  ,\left(
n_{1}^{\prime\prime},k_{1}^{\prime\prime}\right)  ,\ldots,\left(  n_{s}%
^{\prime\prime},k_{s}^{\prime\prime}\right)  \right]  \text{.}%
$$
Then, $r=s$ and, up to a permutation, $\mathrm{profile}\left(  \mathcal{C}%
^{\prime}\right)  =\mathrm{profile}\left(  \mathcal{C}^{\prime\prime}\right)
$, that is, there is $\sigma\in S_{r}$ such that $\left(  n_{i}^{\prime}%
,k_{i}^{\prime}\right)  =(  n_{\sigma\left(  i\right)  }^{\prime\prime
},k_{\sigma\left(  i\right)  }^{\prime\prime})  $ and $\left(
n_{0}^{\prime},k\,_{0}^{\prime}\right)  =\left(  n_{0}^{\prime\prime}%
,k_{0}^{\prime\prime}\right)  $.
\end{theorem}

\begin{proof}
Let $\mathcal{C}^{\prime}=\left(  C^{\prime};C_{0}^{\prime
};C_{i}^{\prime}\right)  _{i=1}^{r}$ and $\mathcal{C}^{\prime\prime}=\left(
C^{\prime\prime};C_{0}^{\prime\prime};C_{i}^{\prime\prime}\right)  _{i=1}^{s}$
be two maximal $P$-decompositions of $C$. Suppose without loss of generality $r<s$. If $T\in GL_{P}\left(  n\right)  $
is an isometry satisfying $T\left(  C^{\prime}\right)  =C^{\prime\prime}$, there is a
component $C_{i}^{\prime}$ of $C^{\prime}$ such that $T\left(  C_{i}^{\prime
}\right)  $ is not contained in any component $C_{j}^{\prime\prime}$ of
$C^{\prime\prime}$, that is, there are $C_{i_{0}}^{\prime},C_{j_{0}}%
^{\prime\prime},C_{j_{1}}^{\prime\prime},\ldots,C_{j_{t}}^{\prime\prime}$ such
that
\[
T\left(  C_{i_{0}}^{\prime}\right)  \subset C_{j_{0}}^{\prime\prime}\oplus
C_{j_{1}}^{\prime\prime}\oplus\cdots\oplus C_{j_{t}}^{\prime\prime}%
\]
but $T\left(  C_{i_{0}}^{\prime}\right) \cap C_{j_{l}}^{\prime\prime} \neq \emptyset$
for any $l$. Here we are assuming $t$ to be strictly positive and minimal with this property. But
the minimality of $t$ implies that%
\[
T\left(  C_{i_{0}}^{\prime}\right)  =\bigoplus_{m=0}^t T\left(  C_{i_{0}}^{\prime
}\right)  \cap C_{j_{m}}^{\prime\prime},
\]
contradicting the hypothesis that  $r<s$. It follows that $r=s$. Moreover,
we have that for every $i\in\{1,\ldots,r\}$ there is $j_{i}$ such that $T\left(
C_{i}^{\prime}\right)  \subseteq C_{j_{i}}^{\prime\prime}$ hence
$n_{i}^{\prime}\leq$ $n_{j_{i}}^{\prime\prime}$ and $k_{i}^{\prime}\leq
k_{j_{i}}^{\prime\prime}$. Applying the same reasoning to $T^{-1}\in
GL_{P}\left(  n\right)  $ we get that $n_{i}^{\prime\prime}\leq$ $n_{j_{i}%
}^{\prime}$and $k_{i}^{\prime\prime}\leq k_{j_{i}}^{\prime}$, hence
$n_{i}^{\prime}=$ $n_{j_{i}}^{\prime\prime}$ and $k_{i}^{\prime}=k_{j_{i}%
}^{\prime\prime}$ so that $\mathrm{profile}\left(  \mathcal{C}^{\prime
}\right)  =\mathrm{profile}\left(  \mathcal{C}^{\prime\prime}\right)  $.
\end{proof}
The next Corollary follows straight from Theorem \ref{teo1}.
\begin{corollary}
Let $\mathcal{C}^{\prime}=\left(  C^{\prime};C_{0}^{\prime};C_{i}^{\prime
}\right)  _{i=1}^{r}$ and $\mathcal{C}^{\prime\prime}=\left(  C^{\prime\prime
};C_{0}^{\prime\prime};C_{i}^{\prime\prime}\right)  _{i=1}^{r}$ be two maximal
$P$-decompositions of $C$ and let $T\in GL_{P}\left(  n\right)  $ be a linear
isometry such that $T\left(  C^{\prime}\right)  =C^{\prime\prime}$. Then,
there is a permutation $\sigma\in S_{r}$ such that $T\left(  C_i^{\prime}\right)  =C_{\sigma\left(  i\right)  }^{\prime\prime}$.
\end{corollary}

\subsection{Primary Decompositions and Complexity for Syndrome Decoding}

Syndrome decoding can be done for poset metrics in the same way it is performed in the Hamming case, just exchanging syndrome leaders to be vectors with minimal poset weight, instead of the Hamming weight. A look-up table for syndrome decoding of $\left(  C^{\prime};C_{0}^{\prime
};C_{i}^{\prime}\right)_{i=1}^{r}$ will have $\sum_{i=1}%
^{r}q^{n_{i}-k_{i}}$ elements where $n_i$ and $k_i$ are obtained by the profile of the decomposition (we are ignoring the component $V_{0}$ since we
know that any codeword $c\in C$ should have $c_{j}=0$ for $j\in\left[
n\right]  ^{C}$). Using these observations, we can define the complexity of a given decomposition as follows:

\begin{definition}
We call 
$$\mathcal{O}\left(  \mathcal{C}^{\prime}\right)
=\sum_{i=1}^{r}q^{n_{i}-k_{i}}$$ the \emph{complexity of the decomposition}.
\end{definition}

\begin{definition}
We say that a $P$-decomposition $\mathcal{C}^{\prime}=\left(  C^{\prime
};C_{0}^{\prime};C_{i}^{\prime}\right)_{i=1}^r$ of $\emph{C}$ is \emph{primary} if $\mathcal{O}\left(  \mathcal{C}^{\prime}\right)$ is minimal among all $P$-decompositions of $C$. This \emph{minimal
complexity of }$C$ \emph{relative to} $P$ is denoted by $\mathcal{O}%
_{P}\left(  C\right)$.
\end{definition}

We will omit the proof, but it is not so difficult to see (due to the previous definitions) that the permutation part $S_{P}$ of $GL_{P}\left(  n\right)  $ is irrelevant to the decomposition of a code.

\begin{lemma}
\label{lema1}Let $\mathcal{C}^{\prime}\mathcal{=}\left(  C^{\prime}%
;C_{0}^{\prime};C_{i}^{\prime}\right)_{i=1}^r$ be a primary
$P$-decomposition of $C$. Let $\sigma\in Aut\left(  P\right)  $ and $T_{\sigma}\in
S_{P}$ be the corresponding isometry. Then
\[
T_{\sigma}\left(  \mathcal{C}^{\prime}\right)  =\left(  T_{\sigma}\left(
C^{\prime}\right)  ;T_{\sigma}\left(  C_{0}^{\prime}\right)  ;T_{\sigma}\left(  C_{i}^{\prime}\right)  \right)_{i=1}^r
\]
is also a primary $P$-decomposition of $C$.
\end{lemma}

We recall that the set $\mathcal{P}_{n}$ of poset over $\left[  n\right]  $ is
itself a partially ordered set and primary decompositions ``behaves well''
according to this order, in the following way: 

\begin{theorem}
\label{conj2}Let $P,Q\in\mathcal{P}_n$ with
$P\leq Q$. Given a code $C$, there is a primary $P$-decomposition of $C$ which is also a $Q$-decomposition of $C$.
\end{theorem}

\begin{proof}
Assume $P,Q\in\mathcal{P}_n  $ and $P<Q$. Let
$\mathcal{C}^{\prime}\mathcal{=}\left(  C^{\prime};C_{0}^{\prime}%
;C_{i}^{\prime}\right)_{i=1}^r  $ be a primary $P$-decomposition
of $C$. Let $\phi\in GL_{P}\left(  n\right)  $ be such that $\phi\left(
C\right)  =C^{\prime}$. We decompose $\phi$ as a product (composition)
$T_{\sigma}\circ$ $A$ of a triangular matrix $A$ and a permutation matrix
$T_{\sigma}$. 
From Lemma \ref{lema1}, we have that
\begin{align*}
\mathcal{C}^{\prime\prime} &  :=T_{\sigma^{-1}}\left(  \mathcal{C}^{\prime
}\right)  \\
&  =\left(  T_{\sigma^{-1}}\left(  C^{\prime}\right)  ;T_{\sigma^{-1}}\left(
C_{0}^{\prime}\right)  ;T_{\sigma^{-1}}\left(  C_{i}^{\prime}\right)  \right)_{i=1}^r  \\
&  =\left(  C^{\prime\prime};C_{0}^{\prime\prime};C_{i}^{\prime\prime
}\right)_{i=1}^r
\end{align*}
is a primary $P$-decomposition of $C$. But
$$
C^{\prime\prime}=T_{\sigma^{-1}}\left(  C^{\prime}\right)   =T_{\sigma^{-1}}\left(
T_{\sigma}  ( C)  \right)  =  C,
$$
therefore $\mathcal{C}^{\prime\prime}$ is a $Q$-decomposition of $C$.
\end{proof}

\begin{corollary}
Let $P,Q\in\mathcal{P}_n  $ with
$P\leq Q$. Then, 
$\mathcal{O}_Q(C)\leq \mathcal{O}_P(C)$
for every $[n,k]_q$ linear code $C$.
\end{corollary}

The situation above is quite general: in what concern primary decompositions,
there are no futile posets, in the sense there is always a code that can be refined when passing from $P$ to $Q$.  This is what is stated in the next proposition, its proof is omitted since it is a long and technical construction, considering many different situations.

\begin{proposition}
\label{conj3}Let $P,Q\in\mathcal{P}_n  $
with $P<Q$. Then there is a code $C$ such that the primary $Q$-decomposition
of $C$ is strictly finer than the primary $P$-decomposition.
\end{proposition}

\subsection{Hierarchical Bounds for Complexity of Syndrome Decoding}

Hierarchical poset metrics are well understood. In particular, if $P$ is a hierarchical poset, the profile of a primary $P$-decomposition of a code $C$ is uniquely (and easily) determined by the weight hierarchy of the code (see \cite{felix} for details). Moreover, it is possible to prove that this property is exclusive of hierarchical posets. For this reason, when considering a general poset $P$, we aim to establish bounds for $\mathcal{O}_P(C)$ considering the easy-to-compute minimal complexity of $C$ relatively to hierarchical posets.

Let $P=\left(  \left[  n\right]  ,\preceq\right)$ be a poset with level
hierarchy $\left[  n\right]  =H_{1}\cup H_{2}\cup\ldots\cup H_{r}$ with
$n_{i}=\left\vert H_{i}\right\vert $ and $n=\sum_{i=1}^r n_{i}$. We say that $P$ is
\emph{hierarchical at level }$i$ if levels $H_{i}$ and $H_{i-1}$ relate
hierarchically, that is, if $a\in H_{i}$ and $b\in H_{i-1}$ then $b\preceq a$. Since $H_0$ is empty, by vacuity every poset is hierarchical at level $1$.
Let $\mathcal{H}\left(  P\right)  =\left\{  i\in\left[  r\right]  |P\text{ is
hierarchical at }i\right\}  $. Of course, $P$ is hierarchical \emph{iff}
$\mathcal{H}\left(  P\right)  =\left[  r\right]  $. We can label the elements
of $\mathcal{H}\left(  P\right)  $ as $1=j_{1}<j_{2}<\cdots<j_{h}$ where
$h=\left\vert \mathcal{H}\left(  P\right)  \right\vert $. 

Out of $P$ we define two hierarchical posets:

\textbf{Upper neighbour:} Let $P^{+}$ be the poset on $\left[n\right]  $
with the same level decomposition of $P$ and for every $a\in H_{i},b\in H_{j}$ with $a\neq b$ define $a\preceq_{P^+}b$ if, and only if, $i<j$.

\textbf{Lower neighbour: }Let $P^{-}$ be the poset on $\left[  n\right]  $ with level  hierarchy $\left[  n\right]  =J_{1}\cup
J_{2}\cup\cdots\cup J_{h}$ where
\begin{equation*}
J_{i}   = \bigcup_{s=j_i}^{j_{i+1}-1} H_s  \ \text{ and } \ j_{h+1}=r+1
\end{equation*}
and for $a\in H_{i},b\in H_{j}$ with $a\neq b$ we have $a\preceq_{P^-}b$ if, and only if, $ i<j$.
It is easy to see that:

\begin{itemize}
\item $P^{+}$ and $P^{-}$ are hierarchical posets and $P$ is hierarchical if, and only if,  $P=P^{+}=P^{-}$;

\item Considering the natural order $\leq$ on $\mathcal{P}_{n}$ we have that
$P^{-}\leq P\leq P^{+}$, moreover 
\begin{align*}
P^{-}  &  =\max\left\{  Q\in\mathcal{P}_{n}|Q\leq P,Q\text{ hierarchical}%
\right\}  ,\\
P^{+}  &  =\min\left\{  Q\in\mathcal{P}_{n}|P\leq Q,Q\text{ hierarchical}%
\right\}.
\end{align*}
\end{itemize}
The proposition below follows directly from the definitions.

\begin{proposition}\label{bound} For any linear code $C$ we have%
\[
 \mathcal{O}_{P^+}\left(  C\right)  \leq \mathcal{O}_{P}\left(  C\right) \leq  \mathcal{O}_{P^-}\left(  C\right).
\]
\end{proposition}

We remark that due to the Canonical-systematic form, introduced in \cite{felix}, the minimal complexity of syndrome decoding is easy to compute for hierarchical metrics.

 If $P$ is not hierarchical, both inequalities are strict, for some code
$C$. Moreover, the bounds are tight, in the sense that, given $P$, there are codes $C^{+}$ and $C^{-}$ such that $\mathcal{O}_{P^{+}}\left(
C^{+}\right)  =\mathcal{O}_{P}\left(  C^{+}\right)  $ and  $\mathcal{O}_{P^{-}%
}\left(  C^{-}\right)  =\mathcal{O}_{P}\left(  C^{-}\right)  $.

\section{Conclusion}
Using the primary $P$-decomposition of a code, we could compare the syndrome decoding complexity of a code to the complexity of syndrome decoding considering hierarchical posets, the unique family of posets where coding invariants are extensively studied. We expect that the same kind of techniques may be used to establish bounds for other difficult problems. For example, it was proved in \cite{rafael} that the packing radius problem of a single pair of vectors is equivalent to the partitioning problem which is know to be a NP-hard problem. Also in this case, hierarchical posets is the family of poset for which this problems is solved and we expect to establish bounds in a way similar to Proposition \ref{bound}.

\section*{Acknowledgment}

The first author was partially supported by FAPESP, grant 2013/25.977-7.
The second author is supported by CNPq. 

\bibliographystyle{IEEEtran}

\bibliography{IEEEabrv,biliog}

\end{document}